\documentclass[10pt,aps,amsfonts,pra,twocolumn,showpacs]{revtex4-1}
\usepackage{verbatim,epsfig,amsmath,amssymb,bm,epsf,graphicx,psfrag,bbold,amsthm,amsfonts}
\usepackage{hyperref}
\usepackage[all]{xy}
\usepackage{color}
\newtheorem{theorem}{Theorem}
\newtheorem{lemma}{Lemma}
\newcommand{\ket}[1]{|#1\rangle}
\newcommand{\bra}[1]{\langle #1|}
\newcommand{\innerproduct}[2]{\langle #1| #2\rangle}
\newcommand{\outerproduct}[2]{|#1\rangle\langle #2|}
\newcommand{\dagr}[1]{{#1}^{\dagger}}
\newcommand{\ztzt}{\mathbb{Z}_2 \times \mathbb{Z}_2}

\newcommand{\af}{A_4}
\newcommand{\hgc}{H^2(G,\mathbb{C})}
\newcommand{\beq}{\begin{equation}}
\newcommand{\eeq}{\end{equation}}
\newcommand{\bea}{\begin{eqnarray}}
\newcommand{\eea}{\end{eqnarray}}

\begin{document}

\title{Ground states of 1D symmetry-protected topological phases and their utility as resource states for quantum computation}

\author{Abhishodh Prakash}
\affiliation{C. N. Yang Institute for Theoretical Physics and
Department of Physics and Astronomy, State University of New York at
Stony Brook, Stony Brook, NY 11794-3840, USA}

\author{Tzu-Chieh Wei}
\affiliation{C. N. Yang Institute for Theoretical Physics and
Department of Physics and Astronomy, State University of New York at
Stony Brook, Stony Brook, NY 11794-3840, USA}
\date{\today}

\begin{abstract}
	The program of classifying symmetry protected topological (SPT) phases in 1D has been recently completed and has opened the doors to study closely the properties of systems belonging to these phases. It was recently found that being able to constrain the form of ground states of SPT order based on symmetry properties also  allows to explore novel resource states for processing of quantum information. In this paper, we generalize the consideration of Else et al.  [Phys. Rev. Lett. {\bf 108}, 240505 (2012)] where it was shown that the ground-state form  of spin-1 chains protected by $\ztzt$ symmetry supports perfect operation of the identity gate, important also for long-distance transmission of quantum information. We develop a formalism to constrain the ground-state form of SPT phases protected by any arbitrary finite symmetry group and use it to examine examples of ground states of SPT phases protected by various finite groups for similar gate protections. We construct a particular Hamiltonian invariant under $A_4$ symmetry transformation which is one of the groups that allows protected identity operation and examine its ground states. We find that there is an extended region where the ground state is the AKLT state, which not only supports the identity gate but also arbitrary single-qubit gates.
\end{abstract}
 \maketitle

 \section{Introduction}
 Symmetry-protected topological (SPT) phases have topological order that is not characterized by a local order parameter and their existence requires symmetry to be preserved~\cite{top1d_wenold, top1d_wennew,top1d_lukasz,top1d_pollmanturner,top1d_schuchgarciacirac}. Ground states of topologically non-trivial SPT phases cannot be continuously connected to trivial product states without either closing the gap or breaking the protected symmetry.  In one dimension, a particularly useful way to describe ground states is the matrix-product-state (MPS) representation~\cite{fannes1992,Perez-Garcia2007,hastings2007area} and this has led to many interesting results including a complete classification of SPT phases~\cite{top1d_wennew}. In addition to classifying SPT phases, an intriguing connection of SPT phases to quantum computation was identified in Ref.~\cite{bartlett2012} that SPT ground states of $\ztzt$ symmetry can serve as resource states for realizing certain gate operations in quantum computation by local measurement. 
 
 Measurement-based quantum computation  (MBQC)~\cite{raussendorf_2001,raussendorf2003,raussendorf_2009} is a quantum computational scheme that makes use of only local measurements on a suitably entangled resource state. It was originally invented with a specific resource state, i.e., the cluster state~\cite{raussendorf_2001} but was subsequently shown to be supported by a variety of systems~\cite{VandenNest2006,gross2007_2,gross2007,XieChen2009}, in particular, the Affleck-Kennedy-Lieb-Tasaki (AKLT) states~\cite{aklt,WeiRobert_Review} on various one- and two-dimensional systems~\cite{miyake_aklt,Wei2011,Miyake2011,wei_aklt2012, wei_aklt2013,Cai_Miyake_2010}. In Ref.~\cite{bartlett2012} it was observed that both the 1D cluster and AKLT states, which are capable of supporting arbitrary single-qubit gates, belong to a 1D SPT phase protected by $\ztzt$ on-site symmetry. Moreover other ground states of this phase also support a protected identity gate operation and can act as perfect wires for transmission of quantum information. The results in Ref.~\cite{bartlett2012} hinge on features of specific Abelian groups, $i.e.$,  groups whose projective representation possesses a maximally non-commutative factor system. This brings forth several interesting questions:
 \begin{enumerate}
 	\item Can we extend the results of Ref~\cite{bartlett2012} to get the ground-state form of SPT phases protected by an arbitrary group (both Abelian and non-Abelian)?
 	\item Are there SPT phases protected by other groups which protect the perfect operation of the identity gate? 
 	\item Are there SPT phases where other non-trivial operations are also allowed? Is it possible to find an entire SPT phase whose ground states support universal one-qubit gates? 
 \end{enumerate}

  Here we develop a formalism that addresses (1) and allows us to treat an arbitrary finite group $G$, either Abelian or non-Abelian, so that we can examine the associated SPT ground states and protected gate operations. The results of Ref.~\cite{bartlett2012} on the spin-1 system with  $\ztzt$  are reproduced in this formulation. To address (2), we find that in addition to $\ztzt$, 1D topologically non-trivial SPT phases associated with the symmetry groups $A_4$ (the alternating group of degree 4) and $S_4$ (the symmetric group of degree 4), see Sec~\ref{sec:examples}, acting on a three-dimensional on-site irreducible representation (i.e., physical spin-1 entities) also protect the identity gate operation. The latter group was also studied in Ref.~\cite{MillerMiyake14}. 
 
 We only make partial progress in answering (3). We consider an example Hamiltonian with $A_4$ and parity invariance and study its ground states in various parameter regimes. This Hamiltonian can be regarded as perturbing the AKLT Hamiltonian. We find an extended region in the parameter space where the ground state is exactly the AKLT state and hence can be used as a resource state capable of universal single-qubit gate operations. Whether or not it is generic that the imposition of an appropriate set of symmetries can allow the entire region of an SPT phase to support protected universal single-qubit gates remains an open question. There has however been progress in reducing certain SPT ground states into resource states that support universal single-qubit operations by a `buffering' technique~\cite{MillerMiyake14}, which in some sense gives an affirmative answer to (3).

The rest of the paper is organized as follows. In Sec~\ref{sec:review}, we review the matrix-product state formalism and its connection to quantum computation, and their utility in SPT phases. In Sec~\ref{sec:mainresults}, we present the key results of our formalism that can determine, in terms of MPS, the structure of SPT ground states constrained by symmetry. The method we used was inspired by Refs.~\cite{vidaltensor,vidaltensor_u1,vidaltensor_su2,vidaltensor_bonddimension} where they consider imposing global symmetries such as $SU(2)$ and $U(1)$ for application in numerical simulations. The formalism we develop here might also find its application in numerical simulations with discrete symmetries imposed~\cite{sukhi_injectivity}. In Sec.~\ref{sec:examples}, we use our formalism to examine SPT phases and their non-trivial ground states protected by symmetries such as $\ztzt$, $D_4$, $A_4$  and $S_4$. In Sec.~\ref{sec:model_hamiltonian}, we construct a specific Hamiltonian that is $A_4$ symmetric by perturbing the AKLT Hamiltonian and study its ground states. We find an extended region where the ground states are identical to the AKLT state, which allows universal single qubit operations. We conclude in Sec.~\ref{sec:summary}.

\section{Review of relevant definitions and results}\label{sec:review}
\subsection{Definition of a gapped phase of matter}
In Ref.~\cite{top1d_wenold}, it was argued that in order to talk about phases of matter, we need to specify the class of Hamiltonians we are considering. Two gapped Hamiltonians from a given class are in the same phase if we can `connect' them smoothly without closing the spectral gap. Otherwise, there is a boundary in the space of Hamiltonians where the gap closes separating different phases of matter~\cite{top1d_wenold,top1d_wennew}. In 1D, if we consider the class of all gapped local Hamiltonians, it has been shown~\cite{top1d_wenold} that they all belong to the same phase and we can connect any two such Hamiltonians without closing the gap by adding suitable local operators. Thus, there is no \emph{intrinsic} topological order in 1D and all Hamiltonians can be connected to those in the trivial phase with product ground states. In other words, any ground state can be connected to a product state. On the other hand, if we restrict ourselves to a class of Hamiltonians that respect some global symmetry, there are generally phase boundaries which arise. We cannot connect Hamiltonians in different phases through symmetry respecting operators without closing the gap. Different phases are characterized by a combination of \emph{symmetry fractionalization} and \emph{symmetry breaking}~\cite{top1d_wennew}. When symmetry is not broken, the unique ground states of these \emph{Symmetry Protected Topological} (SPT) phases respect the symmetry of the Hamiltonian and allow us to write down their form using tools from the representation theory of groups. Much of this is possible by using the matrix-product-state representation of gapped ground states of 1D spin chains which we shall briefly review below.

\subsection{Matrix product states}\label{sec:MPSreview}
We begin by  giving a brief review of the \emph{Matrix Product State} representation of many-body wavefunctions in 1D~\cite{Perez-Garcia2007}. Consider a one-dimensional chain of $N$ spins. If the Hilbert space of each spin is $d$-dimensional, the Hilbert space of the spin chain itself is $d^N$-dimensional. This means that the number of coefficients needed to describe the wavefunction of the spin chain grows exponentially with the length of the chain. However, if the spin chain is in the ground-state configuration of a gapped Hamiltonian, it can be efficiently written as an MPS wavefunction~\cite{vidal03,hastings2007area,arad_itai}. To do this, we need to associate for every spin site  (labeled by $m=1 \dots N$), a $D_m \times D_{m+1}$-dimensional matrix $A^{i_m}_m$ for each basis state $\ket{i_m} = \ket{1} \dots \ket{d}$. $D = max_m(D_m)$ is the maximum `virtual' or `bond' dimension and approaches a constant value that is independent of the size of the chain for gapped spin chains~\cite{hastings2007area}. With these matrices (which we shall refer to as MPS matrices), we can write the wavefunction with periodic boundary conditions as:
\begin{equation} \label{eq:mps_periodic}
  \ket{\psi} = \sum_{i_1 \dots i_N} Tr[ A^{i_1}_1 A^{i_2}_2 \dots  A^{i_N}_N] \ket{i_1} \dots \ket{i_N}.
\end{equation}
We can also write down the wavefunction for a finite chain as
\begin{equation} \label{eq:mps_open}
\ket{\psi} = \sum_{i_1 \dots i_N} \bra{L}  A^{i_1}_1 A^{i_2}_2 \dots  A^{i_N}_N \ket{R} \ket{i_1} \dots \ket{i_N},
\end{equation}
where, the vectors $\ket{L}$ and $\ket{R}$ live in the virtual space and encode the boundary conditions for the finite chain. If we consider the class of local gapped Hamiltonians without any symmetry constraint, Eqs.~(\ref{eq:mps_periodic},\ref{eq:mps_open}) would represent the general form of ground states. This means we need about $Nd$ matrices to specify the ground state.

\subsection{Matrix product states and measurement-based quantum computation}\label{sec:MPS_MBQC}
To demonstrate the motivation for this work, we first see how we can use MPS wavefunctions for MBQC in the virtual space. Consider encoding quantum information that needs to be processed in one of the virtual boundary vectors of Eq.~(\ref{eq:mps_open}), say $\ket{R}$~\cite{gross2007,gross2007_2,upload1,upload2}. If we perform a projective measurement of the $N$-th spin in some basis $\{ \ket{\phi_N^i} \}$ with the outcome being a projection of the spin onto state $\ket{\phi'_N} \in \{ \ket{\phi^i_N} \}$, we can write the wavefunction of the remaining $N-1$ spins as $\ket{\psi'} = \innerproduct{\phi'_N}{\psi}$ $i.e$
\begin{equation} 
\ket{\psi'} = \sum_{i_1 \dots i_{N-1}} \bra{L}  A^{i_1}_1 A^{i_2}_2 \dots  A^{i_{N-1}}_{N-1}  \ket{R'} \ket{i_1} \dots \ket{i_{N-1}},
\end{equation}
where $\ket{R'}=A'_N \ket{R} $ can be regarded as resulting from $\ket{R}$ undergoing a linear transformation $ A'_N = \sum_{i_N} \innerproduct{\phi'_N}{i_N} A^{i_N}_N$. 

Thus, if we know all the MPS matrices $A^{i_m}_m$ and if these matrices span the space of relevant operations on the virtual vector, we can hope to induce any transformation on the vector by measurement in an appropriate choice of basis. Usually, there is also an overall residual operator which we can account for by adapting subsequent bases of measurement. 

Let us demonstrate this using two translationally invariant canonical resource states. First, the cluster state~\cite{raussendorf2003,WeiRobert_Review} is a $d=2$ spin chain whose wavefunction can be written in terms of $D=2$ MPS matrices:
\begin{eqnarray}
\label{eq:cluster_mps}
A^0 = \begin{pmatrix}
1 & 0 \\
1 & 0
\end{pmatrix},~
A^1 = \begin{pmatrix}
0 & 1 \\
0 & -1
\end{pmatrix}
\end{eqnarray}

Measuring in the $\ket{\pm}=\frac{1}{\sqrt{2}}(\ket{0} \pm \ket{1})$ basis results in the operation $\ket{R}\mapsto H (\sigma_z)^s\ket{R}$ where $s$ labels the measurement outcome and is $0/1$ if the outcome is $\ket{\pm}$ and $H$ is the Hadamard gate $H \equiv \frac{1}{\sqrt{2}}\begin{pmatrix}
1 & 1\\
1 & -1
\end{pmatrix}$. The measurement thus induces the Hadamard operation up to residual operators $(\sigma_x)^s$ as $H (\sigma_z)^s = (-1)^s (\sigma_x)^s H$.

We can induce a different operation, say $R_z(\theta)=e^{-i\theta \sigma_z/2}$ by measuring in the basis $\ket{\phi,\pm}=\frac{1}{\sqrt{2}}(\ket{0}\pm e^{i\phi}\ket{1})$. This results in the operation $\ket{R}\mapsto H (\sigma_z)^s e^{-i\phi \sigma_z/2} \ket{R}$ where $s$ is the measurement outcome which is $0/1$ if the outcome is $\ket{\phi, \pm}$. This is a single-qubit rotation by $\phi$ about the Z axis up to the operator $H (\sigma_z)^s$.

Similarly, we can also perform rotations about the other orthogonal axes and using sequential rotations about different axes by appropriate angles (using, for example, the Euler angle parametrization for rotations), we can perform any arbitrary single-qubit rotation.

The second prominent resource state is the AKLT state~\cite{aklt,miyake_aklt,WeiRobert_Review} which is a spin-1 $(d=3)$ system whose wavefunction can be described by $D=2$ MPS matrices:
\begin{equation}
A^i = \sigma_i~~ (i=x,y,z),
\end{equation}
where the basis of the spins $\{\ket{x}, \ket{y}, \ket{z}\}$ is chosen as 
\begin{eqnarray}
\ket{x} \equiv \frac{1}{\sqrt{2}} (\ket{-1} - \ket{1}), \ket{y} \equiv \frac{i}{\sqrt{2}} (\ket{-1} + \ket{1}), \ket{z} \equiv\ket{0}, \nonumber
\end{eqnarray}
with $\ket{\pm1}$ and $\ket{0}$ being eigenstates of the spin-1 $S_z$ operator.
If we measure the spin in \{$\ket{x}$, $\ket{y}$, $\ket{z}$\} basis, we can induce the operation $\ket{R} \mapsto \sigma_s \ket{R}$ which is the identity operation up to the residual operator $\sigma_s$. We can also induce $R_z(\theta) = e^{-i\theta \sigma_z/2}$ by measuring in the basis $\{\ket{\theta,x} = \cos(\frac{\theta}{2}) \ket{x} - \sin(\frac{\theta}{2}) \ket{y}, \ket{\theta,y}= \sin(\frac{\theta}{2}) \ket{x} + \cos(\frac{\theta}{2}) \ket{y}, \ket{z}\}$. If the measurement outcome is $\ket{z}$ then we have the identity operation with residual operator $\sigma_z$. However, if the outcome is $\ket{\theta,x}$ or $\ket{\theta,y}$ then the operation is $\ket{R} \mapsto \sigma_i e^{(-i\theta \sigma_z/2)} \ket{R}$ where $i = x/y$ if the outcome is $\ket{\theta,i}$. Thus, if we keep measuring till we get either $\ket{\theta,x}$ or $\ket{\theta,y}$ as the outcome, we can induce the required operation up to Pauli residual operators. The extension to rotations about other axes and ultimately to a full set of single-qubit rotations is straightforward. An important difference between the AKLT and cluster states is that for the latter, the length of the spin chain needed for computation is fixed while for the former, it is not.

It was noted that both the 1D AKLT and cluster states belong to a non-trivial topological phase protected by $\ztzt$ symmetry~\cite{wen_Tensor_Entanglement,top1d_pollmanturner} and there have been investigations to see if the ability to support quantum computation can be a property of the phase~\cite{miyake_MBQC_edge,bartlett_haldane,bartlett2012,bartlett_njp2012,bartlett_correlation}. In particular, the authors of~\cite{bartlett2012,bartlett_njp2012} deduce that any non-trivial MPS ground state in the non-trivial $\ztzt$ invariant spin-1 Hamiltonians (Haldane phase) must have the form $A^i = B_i \otimes \sigma_i$ $(i=x,y,z)$. Thus, there always exists a `protected' two-dimensional virtual subspace in the ground states of the Haldane phase on which the Pauli matrices act and in which quantum information can in principle be encoded and processed. While the ground states of the Haldane phase in general do not support non-trivial gate operations, they do allow a protected identity gate operation by measurements in the $\{\ket{x},\ket{y},\ket{z} \}$ basis that only induces Pauli operation on the boundary vectors. 

\section{Main result: Tensor decomposed ground-state form in the presence of a global symmetry}\label{sec:mainresults}
\subsection{SPT phases with an on-site internal symmetry }\label{sec:on-site}
Let us now consider symmetric phases of Hamiltonians that are invariant under the action of a certain symmetry group $G$ on each spin according to some representation $u(g)$. i.e. $[H,\hat{U}(g)]=0$ where $\hat{U}(g) = u_1(g) \otimes  \cdots \otimes u_N(g) $. We consider ground states that do not break the symmetry of the Hamiltonian and are hence left invariant under the transformation $\hat{U}(g)$ up to a complex phase
\begin{equation} \label{eq:wavefunction_invariance}
\hat{U}(g) \ket{\psi} = \chi(g)^N \ket{\psi}.
\end{equation}
Eq.~(\ref{eq:wavefunction_invariance}) can be imposed as a condition on the MPS matrix level (Suppressing the site labels for brevity) as~\cite{top1d_wenold, top1d_wennew, top1d_lukasz,top1d_pollmanturner}
\begin{equation}
  \label{eq:mps_invariance}
  u(g)_{ij} A^j = \chi(g)V^{-1}(g) A^i V(g).
\end{equation}
Note that here and henceforth, when no confusion will arise, we use the Einstein summation convention wherein repeated indices are summed over. Because $u$ is a group representation, group properties impose $\chi$ to be a 1D representation and $V$ to be a \emph{projective representation} of $G$. A projective representation respects group multiplication up to an overall complex phase. 
\begin{equation}\label{eq:projective_defn}
V(g_1) V(g_2) = \omega(g_1,g_2) V(g_1 g_2).
\end{equation}
The complex phases $\omega(g_1,g_2)$ are constrained by associativity of group action and fall into classes labelled by the elements of the second cohomology group of $G$ over complex numbers $H^2(G,\mathbb{C})$ (See Appendix~\ref{app:proj} for some comments on projective representations). In other words, the different elements of $H^2(G,\mathbb{C})$ label different classes of projective representations. It was also shown in~\cite{top1d_wenold,top1d_wennew,top1d_lukasz,top1d_pollmanturner,top1d_schuchgarciacirac} that the different elements of $H^2(G,\mathbb{C})$ represent different SPT phases of matter. In particular, the identity element labels the set of \emph{linear representations} of $G$ (which respect group multiplication exactly) and the corresponding phase of matter is trivial, containing product ground states. We now use the symmetry constraint of Eq.~(\ref{eq:mps_invariance}) to deduce the form of the MPS matrices for a given phase labelled by $\omega \in H^2(G,\mathbb{C})$ using a technique similar to the one presented in~\cite{vidaltensor}.

With only on-site symmetry, the different 1D representations $\chi$ all correspond to the same SPT phase~\cite{top1d_wenold,top1d_wennew}. Hence, we just consider the case when $\chi(g) = 1$ $i.e.$ the trivial 1D irreducible representation (irrep) of $G$. With this, we can rewrite Eq.~(\ref{eq:mps_invariance}) in a more illuminating form:
\begin{equation}
\label{eq:invariant_tensor}
u(g)_{ii'} V(g)_{\alpha \alpha'}V^{-1}(g)_{\beta' \beta} A^{i'}_{\alpha' \beta'} = A^{i}_{\alpha \beta}.
\end{equation}
Eq.~(\ref{eq:invariant_tensor}) shows that the matrices $A^i$ are invariant 3 index tensors. We now organize the vector space of each index as a reduced representation constructed out of copies of linear or projective irreps of $G$.
\begin{equation} \label{eq:vectorspace}
\mathbb{V} \cong~ \bigoplus_a n_a \mathbb{V}_a \cong ~\bigoplus_a \mathbb{D}_a \otimes \mathbb{V}_a.
\end{equation}
If $\mathbb{V}$ is the vector space of any index, $a$ runs over the irreps, $n_a$ is the degeneracy (number of copies) of the irrep $a$ and $\mathbb{D}_a$ is the corresponding degeneracy vector space of $a$. Any basis element in the vector space $ \mathbb{V}$ can be labelled by three numbers as $\ket{a_i, m_i, d_i}$ where $a_i$ labels the irreducible representation and is analogous to the angular momentum label in $SU(2)$, $m_i$ labels the state in $a_i$ and is analogous to the azimuthal quantum number $m_i$ and $d_i$ labels which copy of the irreducible representation $a_i$ is being considered. Symmetry transformations are block-diagonal and act on the $m_i$ labels of each sector $a_i$ but leave the $d_i$ labels alone. So if $U(g)$ is a symmetry that acts on the vector space Eq.~(\ref{eq:vectorspace}) and if $U^a(g)$ is the representation of the $a$-th irrep then
\begin{equation}
U(g) \cong~ \bigoplus_{a}\mathbb{1}^a \otimes  U^a(g).
\end{equation}
Note that for a given physical system, we assume that the vector space of the physical index is known in terms of which irreps and how many copies are contained. However, for a given $\omega \in H^2(G,\mathbb{C})$ which labels the phase we are trying to study the ground-state form of, we have to allow an arbitrary number of copies of each projective irrep from the class $\omega$ to appear in the virtual space indices. Using this organization, Eq.~(\ref{eq:invariant_tensor}) and an application of Schur's lemma after decomposing the fusion of the irreps $a_i$ and $a_\alpha$ determined by the Clebsch-Gordan (CG) series $i \otimes \alpha=\oplus_\gamma n_{i \alpha}^\gamma \gamma$ (see Appendix~\ref{app:cg} for more details), we can write down the MPS matrices for the SPT phase labelled by $\omega$ using a generalized Wigner-Eckart theorem as follows 
\begin{multline}
\label{eq:ACB}
A[\omega]^{a_i m_i d_i}_{(a_\alpha m_\alpha d_\alpha) (a_\beta m_\beta d_\beta)} = \\ \sum_{n=1}^{n^\beta_{i \alpha}}  B^{a_i d_i}_{(a_\alpha d_\alpha) (a_\beta d_\beta;n)} C[\omega]^{a_\beta m_\beta;n}_{a_i m_i, a_\alpha m_\alpha},
\end{multline}
where $C[\omega]^{a_\beta m_\beta;n}_{a_i m_i, a_\alpha m_\alpha}$ denotes the CG coefficients associated with the change of basis of the direct product of linear irrep $i$ and the irrep  $\alpha$ of projective class $\omega$, to the $n$-th copy of irrep $\beta$  of the same projective class $\omega$ (See Appendix~\ref{app:cg} for more details) 
\begin{multline}
\label{eq:CG_definition}
\ket{a_\beta, m_\beta ; n} = \\ \sum_{a_i,m_i,a_\alpha, m_\alpha} C[\omega]^{a_\beta m_\beta;n}_{a_i m_i, a_\alpha m_\alpha} \ket{a_i, m_i} \ket{a_\alpha, m_\alpha}.
\end{multline} 
The entries $B^{a_i d_i}_{(a_\alpha d_\alpha) (a_\beta d_\beta;n)}$ of the MPS matrices are not determined by on-site symmetry considerations alone and depend on the parameters of the Hamiltonian amongst other things. Finally, putting back the site dependence, $\mathfrak{m}=1 \cdots N$ in the MPS matrices, we have
\begin{multline}\label{eq:mps_onsite}
A[\omega]^{a_i m_i d_i}_{(a_\alpha m_\alpha d_\alpha) (a_\beta m_\beta d_\beta);\mathfrak{m}} = \\ \sum_{n=1}^{n^\beta_{i \alpha}} B^{a_i d_i}_{(a_\alpha d_\alpha) (a_\beta d_\beta;n); \mathfrak{m}} C[\omega]^{a_\beta m_\beta;n}_{a_i m_i, a_\alpha m_\alpha} ,
\end{multline}

We see that to construct the ground-state form of an SPT phase labelled by $\omega$, we need the CG coefficients for the direct product of the linear representation of the physical spins and the projective irreps of class $\omega$: $\ket{i}$ and $\ket{\alpha}$. To make sense this, we use the result that every finite group $G$ has associated to it at least one other finite group $\tilde{G}$, called a Schur cover, with the property that every projective representation of $G$ can be lifted to a linear representation of $\tilde{G}$~\cite{karpilovsky}. So we can reinterpret the CG coefficients of a linear and projective representation of $G$ simply as the CG coefficients of two linear representations of $\tilde{G}$. For example, half odd integer $j$ representations are projective representations of $SO(3)$ while integer $j$ are linear representations. However, if we consider the group $SU(2)$ which is the cover of $SO(3)$, both half odd integer and integer $j$ are linear representations and we know that we can find CG coefficients for decompositions of the kind $1 \otimes \frac{1}{2} = \frac{1}{2} \oplus \frac{3}{2}$.

To summarize, in order to find the ground-state forms of different SPT phases of a spin chain that transforms under a certain representation $u(g)$ of $G$, we need to follow the following steps:
\begin{enumerate}
\item Obtain the second cohomology group of $G$, $H^2(G,\mathbb{C})$ whose elements $\omega$ will label the different SPT phases.
\item Obtain the covering group $\tilde{G}$ 
\item Identify the irreps `$i$' of the physical spin among the irreps of $\tilde{G}$.
\item Identify the irreps `$\alpha$' that correspond to the projective class $\omega$. 
\item Obtain CG coefficients corresponding to the fusion of the irreps of the physical spin with each irrep of the projective class $\omega$. (Ref.~\cite{sakata} and Appendix~\ref{app:sakata} gives a technique to calculate the CG coefficients for certain types of decompositions of finite group irreps)
\item Use the CG coefficients in Eq.~(\ref{eq:mps_onsite}) allowing $\alpha$ and $\beta$ to run over all the irreps of class $\omega$ and $i$ to run over the irreps of the physical spin. Each block of the MPS matrices split into a part that is calculated purely from the group $G$ for each phase $\omega$ and a part that is undetermined.
\end{enumerate}

\subsection{Obtaining the tensor decomposition of Eq.~(\ref{eq:ACB})}
\label{app:cg}
For what follows, it is useful to employ a basis independent representation of the tensor $A$,
\begin{equation}
\label{eq:basisindeptensor}
\hat{A} = \sum_{i \alpha \beta } A^i_{\alpha \beta} \outerproduct{i \alpha} {\beta}
\end{equation}

We organize the vector space of each index and label it by three quantum numbers--the irrep $a_j$ (analogous to the spin label $j$), the irrep multiplicity $m_j$ (analogous to the azimuthal quantum number $m_j$) and the irrep degeneracy (the number of copies of the irrep, $d_j$), i.e., $\ket{i} = \ket{a_i,m_i,d_i}$, $\ket{\alpha} = \ket{a_\alpha,m_\alpha,d_\alpha}$ and so on. 
\begin{multline}
\!\!\!\!\!\!  \hat{A} = A^{a_i m_i d_i}_{(a_\alpha m_\alpha d_\alpha) (a_\beta m_\beta d_\beta)}
\outerproduct{a_i, m_i,d_i; a_\alpha,m_\alpha,d_\alpha} {a_\beta,m_\beta,d_\beta}. \nonumber
\end{multline}
The invariance condition is 
\beq \label{eq:invariancetensor}
\hat{U}(g) \hat{A} = \hat{A}, 
\eeq
where $\hat{U}(g)$ effects a symmetry transformation on the basis bras and kets of each irrep as
\begin{multline} \label{eq:transformtensor}
\hat{U}(g) \hat{A} \equiv \\  A^{a_i m_i d_i}_{(a_\alpha m_\alpha d_\alpha) (a_\beta m_\beta d_\beta)}   U(g)^i_{m_i m'_i} V(g)^\alpha_{m_\alpha m'_\alpha} V(g)^{-1\beta}_{m'_\beta m_\beta} \\  \outerproduct{a_i,m'_i,d_i; a_\alpha,m'_\alpha,d_\alpha} {a_\beta,m'_\beta,d_\beta} .
\end{multline}
Note that symmetry transformations act on the $m$ indices for each irrep but leave the $d$ indices unchanged. Eqs.~(\ref{eq:invariancetensor}) and~(\ref{eq:transformtensor}) together give us back the tensor invariance condition
\begin{multline}
U(g)^i_{m_i m'_i} V(g)^\alpha_{m_\alpha m'_\alpha} V(g)^{-1\beta}_{m'_\beta m_\beta} A^{a_i m'_i d_i}_{(a_\alpha m'_\alpha d_\alpha) (a_\beta m'_\beta d_\beta)}  \\= A^{a_i m_i d_i}_{(a_\alpha m_\alpha d_\alpha) (a_\beta m_\beta d_\beta)}.
\end{multline}
This condition is valid for each set of irreps labelled by $(a_i, d_i, a_\alpha, d_\alpha, a_\beta, d_\beta)$. Now consider the Clebsch-Gordan (CG) series $i \otimes \alpha=\oplus_\beta n_{i \alpha}^\beta \beta$. On the basis level we have,	
\begin{multline}
\label{eq:cg}
\ket{a_\beta, m_\beta ; n} = \\ \sum_{a_i,m_i,a_\alpha, m_\alpha} C^{a_\beta m_\beta;n}_{a_i m_i, a_\alpha m_\alpha} \ket{a_i, m_i} \ket{a_\alpha, m_\alpha}.
\end{multline} 
$C[\omega]^{a_\beta m_\beta;n}_{a_i m_i a_\alpha m_\alpha}$ denotes the CG coefficients associated with the change of basis of the direct product of irreps $i$ and $\alpha$ to the $n$-th copy of irrep $\beta$. With this, we rewrite Eq.~(\ref{eq:basisindeptensor}) as
\begin{multline}
\hat{A} =  A^{a_i m_i d_i}_{(a_\alpha m_\alpha d_\alpha) (a_\beta m_\beta d_\beta)}  (C^{-1})^{a_\gamma m_\gamma;n}_{a_i m_i, a_\alpha m_\alpha } \\  \outerproduct{a_\gamma, m_\gamma;n, d_i, d_\alpha} {a_\beta,m_\beta,d_\beta}
\end{multline}

The ket $\ket{a_\gamma, m_\gamma;n, d_i, d_\alpha}$ denotes a basis in the $n$-th copy of $a_\gamma$ irrep obtained from fusing the $d_i$-th copy of irrep $a_i$ and $d_\alpha$-th copy of irrep $a_\alpha$.  If we impose invariance Eq.~(\ref{eq:invariancetensor}) in this new form, we get
\begin{multline}
V(g)^{\gamma;n}_{m_\gamma m'_\gamma}  (C^{-1})^{a_\gamma m'_\gamma;n }_{a_i m_i, a_\alpha m_\alpha} A^{a_i m_i
	d_i}_{(a_\alpha m_\alpha d_\alpha) (a_\beta m'_\beta d_\beta)}
V(g)^{-1\beta}_{m'_\beta m_\beta}\\ =  (C^{-1})^{a_\gamma m_\gamma;n }_{a_i m_i ,a_\alpha m_\alpha } A^{a_i m_i
	d_i}_{(a_\alpha m_\alpha d_\alpha) (a_\beta m_\beta d_\beta)},
\end{multline}
which is equivalent to
\begin{multline}
V(g)^{\gamma;n}_{m_\gamma m'_\gamma} \left[ (C^{-1})^{a_\gamma m'_\gamma;n }_{a_i m_i, a_\alpha m_\alpha} A^{a_i m_i
	d_i}_{(a_\alpha m_\alpha d_\alpha) (a_\beta m_\beta d_\beta)}\right]
=\\  \left[ (C^{-1})^{a_\gamma m_\gamma;n }_{a_i m_i, a_\alpha m_\alpha} A^{a_i m_i
	d_i}_{(a_\alpha m_\alpha d_\alpha) (a_\beta m'_\beta d_\beta)}\right]  V(g)^{\beta}_{m'_\beta m_\beta}\
\end{multline}

Using Schur's lemmas, we can now determine that 
\begin{eqnarray}
\gamma \neq \beta &:&  (C^{-1})^{a_\gamma m_\gamma;n }_{(a_i m_i) (a_\alpha m_\alpha) } A^{a_i m_i
	d_i}_{(a_\alpha m_\alpha d_\alpha) (a_\beta m_\beta d_\beta)} = 0 \nonumber, \\
\gamma = \beta  &:&  (C^{-1})^{a_\gamma m_\gamma; n }_{(a_i m_i) (a_\alpha m_\alpha)} A^{a_i m_i
	d_i}_{(a_\alpha m_\alpha d_\alpha) (a_\beta m_\beta d_\beta)} \propto \delta_{m_\gamma m_\beta} .\nonumber
\end{eqnarray}
This gives us 
\begin{eqnarray*}
	&&\left[ (C^{-1})^{a_\beta m_\beta;n }_{a_i m_i, a_\alpha m_\alpha} A^{a_i m_i
		d_i}_{(a_\alpha m_\alpha d_\alpha) (a_\beta n_\beta d_\beta)}\right] \\&& = \delta_{m_{\beta} n_{\beta}}~ B^{a_i d_i}_{(a_\alpha d_\alpha) (a_\beta d_\beta;n)}
\end{eqnarray*}
This is again a condition valid for each set of irreps labelled by $(a_i, d_i, a_\alpha, d_\alpha, a_\beta, d_\beta)$. Finally, moving $C$ to the right hand side, we get
\beq
\label{eq:decomposed_mps}
A^{a_i m_i d_i}_{(a_\alpha m_\alpha d_\alpha)(a_\beta m_\beta d_\beta)} = \sum_{n=1}^{n_{i \alpha}^\beta}  B^{a_i d_i}_{(a_\alpha d_\alpha) (a_\beta d_\beta;n)} C^{a_\beta m_\beta;n }_{a_i m_i, a_\alpha m_\alpha}
\eeq
If we restrict $V$ to contain only irreps of a class $\omega$, we get Eq.~(\ref{eq:ACB}).

\subsection{SPT phases with on-site symmetry and lattice translation invariance}\label{sec:on-site+trans}
Gapped Hamiltonians with only lattice translation invariance all belong to the same phase~\cite{top1d_wenold,top1d_wennew}. Ground-states of such Hamiltonians and can be described by MPS matrices $A^{i_m}_m$ that are site independent $i.e.$ $A^{i_m}$~\cite{Perez-Garcia2007}. This means that unlike the case for an arbitrary gapped phase where we needed $N d$ matrices to describe a ground state, we now only need $d$ matrices. Eq.~(\ref{eq:mps_periodic}) is simplified to
\begin{equation} \label{eq:mps_trans}
  \ket{\psi} = \sum_{i_1 \dots i_N} Tr[ A^{i_1} A^{i_2} \dots  A^{i_N}]\ket{i_1} \dots \ket{i_N}.
\end{equation}

If we consider gapped Hamiltonians invariant under translation and an on-site symmetry transformation $u(g)$, the conditions of Eqs.~(\ref{eq:wavefunction_invariance}, \ref{eq:mps_invariance}) again hold. However, unlike the case for just on-site symmetry, the different 1D irreps, $\chi(g)$ that appear in Eq.~(\ref{eq:mps_invariance}) now label distinct phases of matter~\cite{top1d_wenold,top1d_wennew}. Different SPT phases are now labelled by $\{ \omega, \chi \}$ where, $\omega \in H^2(G,\mathbb{C})$ labels the different projective classes and $\chi$ labels the different 1D irreps of the group $G$. We now see how we can constrain the ground-state form of these SPT phases extending the results of Sec~\ref{sec:on-site}

Let us rewrite Eq.~(\ref{eq:mps_invariance}) by absorbing $\chi(g)$ on the right hand side into $u(g)$ on the left and call $\tilde{u}(g) = \chi^*(g) u(g)$
\begin{equation}
  \tilde{u}(g)_{ij} A^j =V^{-1}(g) A^i V(g),
\end{equation}

Since re-phasing a representation with a 1D irrep is still a representation, we can find the new irrep content of $\tilde{u}(g)$. With this, we can repeat the procedure of Sec~\ref{sec:on-site} and obtain the MPS matrices for ground states of a given spin system in any phase labelled by $\{\omega, \chi\}$ as
\begin{multline}\label{eq:mps_onsite+trans}
A[\omega,\chi]^{a_i m_i d_i}_{(a_\alpha m_\alpha d_\alpha) (a_\beta m_\beta d_\beta)} =\\ \sum_{n=1}^{n^\beta_{i \alpha}}  B^{a_i d_i}_{(a_\alpha d_\alpha) (a_\beta d_\beta;n)} C[\omega,\chi]^{a_\beta m_\beta;n}_{a_{i'} m_{i'}, a_\alpha m_\alpha}.
\end{multline}
Where $i \otimes \chi \cong i'$ is some linear irrep of $G$ that can easily be identified by calculating the characters of $i'$ and $C[\omega,\chi]^{a_\beta m_\beta;n}_{a_{i'} m_{i'} a_\alpha m_\alpha}$ denote the CG coefficients associated with the change of basis of the direct product of linear irrep $i'$ and the irrep  $\alpha$ of projective class $\omega$, to the $n$-th copy of irrep $\beta$  of the same projective class $\omega$.

To summarize, in order to find the ground-state forms of different SPT phases for a spin chain that transforms under a certain representation $u(g)$ of $G$  and that is translationally invariant, we need to follow the steps below:
\begin{enumerate}
\item Obtain $H^2(G,\mathbb{C})$ and the covering group $\tilde{G}$.
\item Identify the irreps `$i$' of the physical spin among the irreps of $\tilde{G}$. 
\item Identify the different 1D irreps of $G$, $\chi$ among the 1D irreps of $\tilde{G}$.
\item Identify the irreps `$i'$' corresponding to re-phasing the physical spin irreps `$i$' with $\chi$.
\item Identify which irreps `$\alpha$' correspond to the projective class $\omega$. 
\item Obtain CG coefficients corresponding to the fusion of the re-phased irreps of the physical spin with each irrep of the projective class $\omega$.
\item Use the CG coefficients in Eq.~(\ref{eq:mps_onsite+trans}) allowing $\alpha$ and $\beta$ to run over all the irreps of class $\omega$ and $i'$ to run over the re-phased irreps of the physical spin.
\end{enumerate}

We can also consider the ground-state forms constrained by other space-time symmetries like inversion and time-reversal and combinations with on-site symmetry which have also been classified. While there are constraints imposed on the entries of the MPS matrices, we do not immediately see a useful structure like we do with on-site symmetries with or without translation invariance mentioned above. However, for the sake of completeness, we have presented the results in the Appendices.~\ref{app:parity+time},\ref{app:parity+time+onsite}.

\section{Examples of ground-state forms for various on-site symmetries}\label{sec:examples}
In this section, we use the results of the decomposition scheme discussed in the previous section to write down several ground-state forms of SPT phases protected by various on-site symmetries. 

We will focus on some subgroups of $SO(3)$ that have a particular non-trivial second cohomology group $\hgc = \mathbb{Z}_2$ and hence one class of non-trivial projective representations. (But our formalism can be applied to groups of other second cohomology group as well.) We will also focus on constructing ground states that are topologically non-trivial $i.e.$ states that cannot be connected to the product state and whose virtual space representation corresponds to non-trivial projective representation. This is because these non-trivial states are sufficiently entangled and may offer advantages for information processing. We shall use the following conventions: 
\begin{enumerate}
	\item Groups are defined by a \emph{presentation} $\innerproduct{S}{R}$ $i.e.$ by listing the set $S$ of generators and the set $R$ of relations between them. 
	\item Representations are written by listing those of the generating set $S$. Any element in the group can always be written as the product of powers of the subset of $S$.
	\item $\tilde{G}$ denotes the Schur cover of $G$ that contains the linear and projective irreps of $G$. 
	\item We list the irreps of $\tilde{G}$ and label different classes of irreps by elements of $H^2(G,\mathbb{C})$. These correspond to the linear and projective irreps of $G$.
	\item $\chi_i$ denotes different 1D irreps of $G$ (and $\tilde{G}$).
	\item MPS matrices are constructed up to a similarity transformation for a particular basis of the physical spin that will be mentioned.
	\item Pauli matrices are denoted as $\sigma_i = \{\sigma_x, \sigma_y, \sigma_z\}$ or $\sigma_i = \{\sigma_1, \sigma_2, \sigma_3\}$.
\end{enumerate}
\subsection{Haldane phase ($\ztzt$)}\label{sec:z2z2}
Consider a chain of three level spins (d=3) that is invariant under a three-dimensional representation of $\ztzt$ written as a restricted set of spin-1 $SO(3)$ rotations,
\begin{equation}\label{eq:u_z2z2}
u(g) = \{\mathbb{1},R_x(\pi),R_y(\pi),R_z(\pi)\}.
\end{equation}
$\ztzt$, also known as the Klein four-group, is the group of symmetries of a rhombus or a rectangle (which are not squares) generated by $\pi$ flips about perpendicular axes in the plane of the object. Some information about the group are follows:
\begin{itemize}
	\item $G=\ztzt = \langle a,x | a^2 = x^2 = (ax)^2 = e\rangle $
	\item $H^2(G,\mathbb{C}) = \mathbb{Z}_2 = \{e,a\}$
	\item $\tilde{G}= D_8 : \langle a,x|a^4 = x^2 = (ax)^2 = e \rangle$
	\item Class $e$ irreps of $\tilde{G}$:\\
	$1_{(p,q)}: a \mapsto (-1)^p,~x \mapsto(-1)^q$, $~(p,q)$ $\in$ $\{0,1\}$
	\item Class $a$ irreps of $\tilde{G}$:\\
	$\tilde{2}: a \mapsto i\sigma_z $,~$x \mapsto\sigma_x$
\end{itemize}	
The three-dimensional representation can be shown to be $u(g) \cong 1_{(0,1)} \oplus 1_{(1,0)} \oplus 1_{(1,1)}$. Which means, with an appropriate choice of basis, each basis state of the 3 level spin transforms as one of the non-trivial 1D irreps. We can check that 
$\{\ket{x} \equiv \frac{1}{\sqrt{2}} (\ket{-1}-\ket{1}),~\ket{y} \equiv \frac{i}{\sqrt{2}} (\ket{-1}+\ket{1}), ~\ket{z} \equiv \ket{0}\}$ is such an appropriate basis where $u(g)$ is block diagonal. Calculating the CG coefficients, we get the following MPS matrices:
\begin{equation}
A^i = B_i \otimes \sigma_i, \label{eqn:factorize}
\end{equation}
where $B_i$ are undetermined and $\sigma_i$ are the Pauli matrices. We thus have reproduced the result of Ref.~\cite{bartlett2012} using our general framework.

\subsection{$D_4$ invariant SPT phase}\label{sec:d4}
$D_{n}$, the dihedral group is the symmetry group of a planar $n$ sided polygon and has projective representations when $n$ is even. Some information about the group are as follows. We only look at the case of even $n$.

\begin{enumerate}
	\item G = $D_n$ = $\innerproduct{a,x}{a^n = x^2 = (ax)^2 = e }$
	\item $\hgc = \mathbb{Z}_2 = \{e,a\}$
	\item $\tilde{G}$ = $Q_{n}$: \\$\innerproduct{a,x}{a^{2n} = x^4 = e, a^n = x^2, x a x^{-1 } = a^{-1}}$.
	\item Class $e$ irreps of $\tilde{G}:$
	\begin{enumerate}
		\item 	$1_{(p,q)}: a \mapsto (-1)^p,~x \mapsto (-1)^q$, $~(p,q)$ $\in$ $\{0,1\}$
		\item $2_{(k)}: a \mapsto \begin{pmatrix}
		e^{ {-i k \eta_n}/{2} } & 0\\
		0 & e^{ {i k \eta_n}/{2} },
		\end{pmatrix} $,
		$x \mapsto -i \sigma_y$,\\ $k = 2,4,\ldots n-2$, $\eta_n = 2 \pi/ n$
	\end{enumerate}
	\item Class $a$ irreps of $\tilde{G}:$
	\begin{enumerate}
		\item $\tilde{2}_{(k)}: a \mapsto \begin{pmatrix}
		e^{ {-i k \eta_n}/{2} } & 0\\
		0 & e^{ {i k \eta_n}/{2} }
		\end{pmatrix} $,
		$x \mapsto  \sigma_y~$,\\ $k = 1,3,\ldots n-1$, $\eta_n = 2 \pi/ n$
	\end{enumerate}	
\end{enumerate}
Let us now consider the group $D_4$. This is the group of symmetries of a square generated by $\frac{\pi}{2}$ rotations about the symmetry axis perpendicular to the plane and reflections about symmetry axes in the plane of the square. We consider the following irreps (using a different choice of basis than the one mentioned above).
\begin{enumerate}
	\item Linear irrep $2_{(2)}: a \mapsto i \sigma_y$, $x\mapsto \sigma_z$
	\item Projective irreps:\\ $\tilde{2}_{(1/3)}:a \mapsto \frac{1}{\sqrt{2}} (\pm \mathbb{1}-i \sigma_y),~x \mapsto i \sigma_z$
\end{enumerate}
If we consider a $d=2$ physical spin transforming under the 2D irrep $2_{(2)}$ the non-trivial MPS matrices associated with the two basis states $\ket{i} = \ket{0}, \ket{1}$ are obtained by calculating the CG coefficients:
\begin{eqnarray}
A^0 &=& \begin{pmatrix}
B_{11} \otimes \sigma_z & B_{13} \otimes \mathbb{1}\\
B_{31} \otimes \mathbb{1} & B_{33} \otimes \sigma_z\\
\end{pmatrix}, \\
A^1 &=& \begin{pmatrix}
B_{11} \otimes -\sigma_x & B_{13} \otimes -i \sigma_y\\
B_{31} \otimes i \sigma_y & B_{33} \otimes \sigma_x\\
\end{pmatrix}.
\end{eqnarray}
The MPS matrices cannot be further factorized, and thus we do not even have the protected identity gate. 
\subsection{$A_4$ invariant SPT phase}\label{subsec:a4}
$A_4$, the alternating group of degree four, is the group of chiral or rotational symmetries of a regular tetrahedron generated by rotations (no reflections) about various symmetry axes. It is also the group of even permutations on four elements, i.e. a subgroup of $S_4$ to be discussed next. Some information about the group are as follows.
\begin{enumerate}
	\item G = $A_4$ = $\innerproduct{a,x}{a^3 = x^2 = (ax)^3 = e }$
	\item $\hgc = \mathbb{Z}_2 = \{e,a\}$
	\item $\tilde{G}$ = $\tilde{T}$ : $\innerproduct{a,x}{a^3 = x^2 = v, v^2= (ax)^3 = e}$.
		\item Class $e$ irreps of $\tilde{G}:$
		\begin{enumerate}
			\item $1_{(p)}: a \mapsto e^{2 \pi i p/3},~x \mapsto 1$, $p=0,1,2$
			\item $3: a \mapsto \begin{pmatrix}
			0 & 1 & 0\\
			0 & 0 & 1\\
			1 & 0 & 0
			\end{pmatrix}$, $x \mapsto \begin{pmatrix}
			1 & 0 & 0 \\
			0 & -1 & 0\\
			0 & 0 & -1
			\end{pmatrix}$
		\end{enumerate}
		\item Class $a$ irreps of $\tilde{G}:$
		\begin{enumerate}
			\item $\tilde{2}_{(p)}: a \mapsto e^{2 \pi i p/3} \frac{1}{2} [\mathbb{1} + i (\sigma_x + \sigma_y + \sigma_z)] $,\\$ x \mapsto i \sigma_x $, $p = 0,1,2$
		\end{enumerate}
	\end{enumerate}
	If we consider the physical spin transforming under the only 3D linear irrep, the non-trivial MPS matrices associated with the three basis states $\ket{i} = \ket{1}, \ket{2}, \ket{3}$ are obtained by calculating the CG coefficients:
	\begin{eqnarray}
	A^i &=& B_{i} \otimes \sigma_i\\
	B_{i} &=& V^{i-1} B V^{*i-1}\\
	V &=& \begin{pmatrix}
	\mathbb{1} & 0 & 0\\
	0 & \omega~\mathbb{1} & 0\\
	0 & 0 & \omega^*~\mathbb{1}
	\end{pmatrix},~\omega = e^{2 \pi i /3} \nonumber \\
	B &=& \begin{pmatrix}
	B_{00} & B_{01} & B_{02}\\
	B_{10} & B_{11} & B_{12}\\
	B_{20} & B_{21} & B_{22}\\
	\end{pmatrix}. \nonumber
	\end{eqnarray}
	Similar to Eq.~(\ref{eqn:factorize}) the MPS matrices are factorized into to two parts, and the indenity gate is protected by the symmetry.  We remark that imposing inversion or time-reversal symmetry does not further simplify the $B$'s structure.
\subsection{$S_4$ invariant SPT phase}\label{sec:s4}
$S_4$, the symmetric group of degree four, is the group of achiral or full symmetries of a tetrahedron generated by rotations and reflections about various symmetry axes. It is also the group of all permutations of four elements. Some information about the group are as follows.
\begin{enumerate}
	\item G = $S_4$ = $\innerproduct{a,b,c}{a^2 = b^3 = c^4 = abc = e }$
	\item $\hgc = \mathbb{Z}_2 = \{e,a\}$
	\item $\tilde{G}$ = $O'$ : $\innerproduct{a,b,c}{a^2 = b^3 = c^4 = abc = v, v^2= e }$
	\item Class $e$ irreps of $\tilde{G}:$ $(a= tk,~ b = s,~ c = s^2 k t)$
	\begin{enumerate}
		\item $1_{(p)}:t \mapsto (-1)^p,~ k \mapsto 1,~ s \mapsto 1$, $p= \{0,1\}$
		\item $s: t \mapsto \sigma_x,~ k \mapsto \mathbb{1},~ s\mapsto \begin{pmatrix}
		e^{2 \pi i/3} & 0\\
		0 & e^{-2 \pi i/3}
		\end{pmatrix}$
		\item $3_{(p)}$: $k \mapsto \begin{pmatrix}
			1 & 0 & 0 \\
			0 & -1 & 0\\
			0 & 0 & -1
		\end{pmatrix}$,~$
		s \mapsto \begin{pmatrix}
		0 & 1 & 0\\
		0 & 0 & 1\\
		1 & 0 & 0
		\end{pmatrix}$,\\  $t \mapsto (-1)^p \begin{pmatrix}
		1 & 0 & 0\\
		0 & 0 & 1\\
		0 & 1 & 0
		\end{pmatrix}$, $p = 0,1$
	\end{enumerate}
	\item Class $a$ irreps of $\tilde{G}:$
	\begin{enumerate}
		\item $\tilde{2}_{(p)}:t \mapsto (-1)^p \frac{i}{\sqrt{2}}(\sigma_z - \sigma_y),~ k \mapsto i \sigma_x,\\ s \mapsto \frac{1}{2} [\mathbb{1} + i (\sigma_x + \sigma_y + \sigma_z)]$ , $p =0,1$
		\item $\tilde{4} = 2 \otimes \tilde{2}_{(0)}$
	\end{enumerate}
\end{enumerate}
If we consider the physical spin transforming under one of the 3D linear irreps, $3_{(1)}$ the non-trivial MPS matrices associated with the three basis states are obtained by calculating the CG coefficients:
\begin{eqnarray}
A^i &=& B_{i} \otimes \sigma_i\\
B_{i} &=& \begin{pmatrix}
B_{2_02_0} & 0 & B_{2_0 4} \otimes u^\dagger_{i-1}\\
0 & B_{2_12_1} & B_{2_1 4} \otimes v^\dagger_{i-1}\\
B_{4 2_0} \otimes u_{i-1} & 
B_{4 2_1} \otimes v_{i-1} &
B_{4 4} \otimes \mathbb{1}_2 + \tilde{B}_{44} \otimes f_{i-1}
\end{pmatrix} \nonumber\\
u_i &=& \begin{pmatrix}
\omega^{*i}\\
\omega^{i}
\end{pmatrix},~ 
v_i =  \begin{pmatrix}
\omega^{*i}\\
-\omega^{i}
\end{pmatrix} \nonumber\\
f_i &=&  \begin{pmatrix}
0 & \omega^i\\
\omega^{*i} & 0
\end{pmatrix},~ \omega = e^{2 \pi i/3} \nonumber
\end{eqnarray}
We observe that  if we restrict the $B_i$ matrix to only the bottom right block and set the two matrices to scalars, $B_{4 4} = \cos(\frac{\theta}{2})$ and $\tilde{B}_{4 4} = e^{i\phi} \sin(\frac{\theta}{2})$, then it reduces to the one used for the buffering scheme in Ref.~\cite{MillerMiyake14} up to a change of basis.

\subsection{Summary of new SPT phases with identity gate protection}
We now list, from the examples in the previous section, those SPT ground states which allow the perfect operation of the identity gate according to the scheme reviewed in Sec.~\ref{sec:MPS_MBQC}. We see that the MPS matrices for non-trivial ground states of $d=3$ (i.e. of spin magnitude $S=1$) spin chains protected by $\ztzt,~A_4,~S_4$ all have the form 
\begin{eqnarray} \label{eq:good_form_for_identity}
A^i = B_i \otimes \sigma_i.
\end{eqnarray}
 Following the convention of Ref.~\cite{bartlett2012}, we call $B_i$ the junk part and $\sigma_i$ the protected part. We also note that our convention of placing the protected and junk parts is in reverse order as compared to the convention used in Refs.~\cite{bartlett2012,bartlett_njp2012,MillerMiyake14}. This is for notational consistency in this paper.
  
Consider encoding qubit information $\ket{\psi}$ in the protected part of the right boundary virtual space with the junk part arbitrarily set to some state $\ket{J}$ in any of these ground states.
\begin{eqnarray}
\ket{R} = \ket{J} \otimes \ket{\psi}
\end{eqnarray}

If we perform a measurement on the rightmost $i.e$ $N$-th spin in the basis $\ket{x}, \ket{y}, \ket{z}$ in which the MPS matrices have the form of Eq.~(\ref{eq:good_form_for_identity}) with an outcome $\ket{k_N}$, we induce a transformation of the boundary vector by~(\ref{sec:MPS_MBQC})
\begin{eqnarray}
\ket{R} &\mapsto& A^{k_N} \ket{R}\\
\implies \ket{J} \otimes \ket{\psi} &\mapsto& B_{k_N} \ket{J}\otimes \sigma_{k_N} \ket{\psi} 
\end{eqnarray}
The qubit information $\ket{\psi}$ is unchanged upto an inconsequential Pauli operator $\sigma_{k_N}$ which can be corrected for by a change of readout basis. In fact, we can measure several spins (say $m$ from the right) and we still have the perfect operation of the identity gate upto a residual operator $\sigma_{k_{N-m}} \ldots \sigma_{k_N}$. This means all these ground states allow a protected subspace with perfect identity gate operation, which allows for perfect transmission of quantum information encoded in the projected subspace. 

However, note that if we measure in a different basis formed by a linear combination of $\ket{x}, \ket{y}, \ket{z}$, it is easy to check that the boundary vector $\ket{R} = \ket{J} \otimes \ket{\psi}$ no longer remains decomposed into protected and junk parts and, in general, there will be mixing between the two vector spaces. As an illustration, if a measurement outcome of $\frac{1}{\sqrt{2}}(\ket{x} + \ket{y})$ is obtained, the induced transformation on $\ket{R}$ is (up to an overall factor) 
\begin{eqnarray}
\ket{J} \otimes \ket{\psi} \mapsto B_x \ket{J} \otimes \sigma_x \ket{\psi} + B_y \ket{J} \otimes \sigma_y \ket{\psi}.
\end{eqnarray}
 Thus, in general only the identity gate is protected in the ground states of these phases. However, if it were possible that $B_i$ is independent of physical index $i$, then arbitrary single-qubit gates would be possible, as mixing will not occur. It is worth noting that when $B_i$ is independent of the index $i$, the corresponding wavefunction is identically the AKLT state. We had hoped that imposing additional symmetry like parity and/or time reversal invariance might give further constraints on the matrices $B_i$'s and thereby allow non-trivial gate operations. But we checked (using results of Appendix.~\ref{app:parity+time+onsite}) that imposing these additional symmetries on the $\ztzt, \af$ and $S_4$ SPT ground states listed above does not induce ground states that could provide universal qubit operations.

\section{An $A_4$ symmetric Hamiltonian}
\label{sec:model_hamiltonian}

Here we ask a slightly less general question: can one find a particular Hamiltonian with symmetry such that there is an extended region (not necessarily at all points of a phase) in the phase diagram that the ground states can provide universal qubit operations in the framework of MBQC? 
 Below we first construct a specific Hamiltonian that possesses $A_4$ and parity symmetry, which can be regarded as perturbing the spin-1 AKLT Hamiltonian. 
Then we present a numerical investigation and show that indeed there exists a finite parameter region where the ground states are exactly (here and henceforth, exact is defined up to machine precision) the AKLT state, and can therefore serve as a resource state for implementing universal single-qubit gates. After the numerical investigation, we present analytic understanding why such an extended region of AKLT ground states can exist.

\begin{figure}[ht]
	\includegraphics[scale = 0.63]{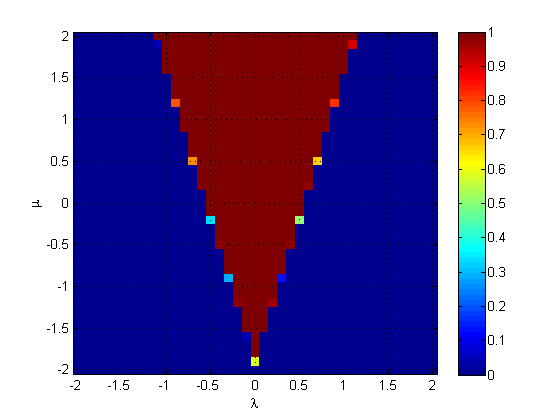}
	\caption{(Color online) Fidelity of ground states with the AKLT state. It is seen that there is an extended region such that the ground state is exactly the AKLT state.}
	\label{fig:Fidelity}
\end{figure}

\subsection{Construction of $A_4$ symmetric Hamiltonian}
We will now construct the $\af$ and inversion symmetric Hamiltonian and study its phase diagram. We use group invariant polynomials as building blocks to construct Hermitian operators invariant under group action. A group $G$ invariant $n$-variable polynomial $f(x_1,x_2, \dots  x_n)$ is unchanged when the $n$-tuplet of variables $\left(x_1, x_2 \dots x_n \right)$ is transformed under an $n$-dimensional representation of the group $D(g)$.
\begin{eqnarray}
f(x'_1, x'_2 \dots x'_n) &=& f(x_1, x_2 \dots x_n) \\
x'_i &=&  D(g)_{i j} x_j.
\end{eqnarray}
If we have $n$ Hermitian operators $X_{i=1 \dots n}$ that are $n$-dimensional and transform covariantly like the $n$ variables of the polynomial $x_{i=1 \dots n}$, $i.e.$ $D(g) X_i \dagr{D}(g) = D(g)_{i j} X_j$, then we can elevate the group invariant polynomials to group invariant operators as $f(x_1,x_2, \dots  x_n) \rightarrow f(X_i, X_2 \dots X_n)$ carefully taking into account that unlike the numbers $x_i$, the operators $X_i$ do not commute. 

Since we need three-dimensional operators of $\af$, we consider the set of independent three variable polynomials invariant under the 3D irrep of $\af$~\cite{ramond_group}:
\begin{eqnarray}
f_1(x,y,z) &=& x^2 + y^2 + z^2,\\
f_2(x,y,z) &=& x^4 + y^4 + z^4,\\
f_3(x,y,z) &=& xyz.
\end{eqnarray}
We know that the spin operators $S^i$ satisfying $ [S^i, S^j] = i \epsilon_{ijk} S^k $ transform covariantly under any $SO(3)$ rotation, in particular for the finite set of rotations that corresponds to the subgroup $\af \in SO(3)$. Thus, to find invariant operators for the three-dimensional representation $3$, we need to take the spin operators in the appropriate three-dimensional basis $\ket{x}, \ket{y}, \ket{z}$ as defined in Sec.~\ref{subsec:a4} and elevate the polynomials $f_1, f_2, f_3$ to operators as
\begin{eqnarray}
F_1 &=& S^x_a S^x_b + S^y_a S^y_b + S^z_a S^z_b,\\
F_2 &=& (S^x_a S^x_b)^2 + (S^y_a S^y_b)^2 + (S^z_a S^z_b)^2 ,\\
F_3 &=& S^x_a S^y_b S^z_c + S^z_a S^x_b S^y_c + S^y_a S^z_b S^x_c \nonumber \\
&+& S^y_a S^x_b S^z_c + S^x_a S^z_b S^y_c + S^z_a S^y_b S^x_c, 
\end{eqnarray}
where the indices $a,b,c$ label collectively any other quantum numbers like lattice sites and can be chosen as per convenience, say to make the operators local. As a model Hamiltonian, we could use any function of the invariant operators $F_1$, $F_2$ and $F_3$ and ensure that everything is symmetric under the exchange of lattice labels to impose inversion symmetry. In particular, the AKLT state is the unique ground state of a particular combination of the invariant operators but it has  a larger symmetry group, $SO(3)$.
\begin{eqnarray}
H_{AKLT} = \sum_i \left[ \vec{S}_i\cdot\vec{S}_{i+1} + \frac{1}{3} (\vec{S}_i\cdot\vec{S}_{i+1})^2 \right], \\
\mbox{where}~\vec{S}_i\cdot\vec{S}_{i+1} \equiv S^x_i S^x_{i+1} + S^y_i S^y_{i+1} + S^z_i S^z_{i+1}. \nonumber
\end{eqnarray}
Thus we can consider adding two other combinations to the AKLT Hamiltonian so as to break the $SO(3)$ symmetry to $\af$ by using $\af$ invariant perturbations:
\begin{eqnarray}
H_q &=& \sum_i \left[ (\vec{S}^2_i\cdot\vec{S}^2_{i+1}) - \frac{1}{3} (\vec{S}_i\cdot\vec{S}_{i+1})^2  \right],\\
\mbox{where}&&\vec{S}^2_i\cdot\vec{S}^2_{i+1} \equiv (S^x_i S^x_{i+1})^2 + (S^y_i S^y_{i+1})^2 + (S^z_i S^z_{i+1})^2,\nonumber
\end{eqnarray}
and
\begin{multline}
H_c = \sum_i [ (S^x S^y)_i S^z_{i+1} + (S^z S^x)_i S^y_{i+1} + (S^y S^z)_i S^x_{i+1} \\
+ (S^y S^x)_i S^z_{i+1} + (S^x S^z)_i S^y_{i+1} + (S^z S^y)_i S^x_{i+1} \\
+ S^x_{i} (S^y S^z)_{i+1} + S^z_{i} (S^x S^y)_{i+1} + S^y_{i} (S^z S^x)_{i+1}  \\
+  S^x_{i} (S^z S^y)_{i+1} + S^z_{i} (S^y S^x)_{i+1} + S^y_{i} (S^x S^z)_{i+1}]. 
\end{multline}
With these pieces, we arrive at the total Hamiltonian which is $\af$ and inversion symmetric,
\begin{equation}
\label{eqn:H}
H = H_{AKLT} + \lambda H_c + \mu H_q.
\end{equation}

\subsection{Checking AKLT as the ground state}
The AKLT state  $|\psi_{\rm AKLT}\rangle$ has the MPS representation  $A^x=\sigma_x$, $A^y=\sigma_y$, and $A^z=\sigma_z$ in the basis of \{$|x\rangle$, $|y\rangle$, $|z\rangle$\} defined earlier. We know that at $\lambda=\mu=0$ the ground state of the Hamiltonian~(\ref{eqn:H}) is uniquely the AKLT state. We would like to know whether there is an extended region of $(\lambda,\mu)$ around $(0,0)$ such that the ground state is also the AKLT state. We do this numerically by first solving the ground state $|\psi_G\rangle$ of the Hamiltonian~(\ref{eqn:H}) using the infinite time-evolving bond decimation (iTEBD) algorithm invented by Vidal~\cite{vidal2007} and then calculating the fidelity between these two states $f=|\langle \psi_G|\psi_{\rm AKLT}\rangle|^2$. As shown in Fig.~\ref{fig:Fidelity} we indeed see that there is an extended region in this Hamiltonian such that the ground state is exactly the AKLT state and thus a useful resource state for universal single-qubit MBQC.  \\

\subsection{Analytic understanding}
We now analyze why such an extended region of AKLT is possible and calculate analytically the boundary of the AKLT region in the $\lambda$-$\mu$ plane, shown in Fig.~\ref{fig:Fidelity}. First we recall that the interaction between sites $i$ and $i+1$ of ${H_{AKLT}}$ is a projection to the joint $S=2$ subspace. More precisely,
\begin{equation}
(H_{AKLT})_{i,i+1}= 2 \sum_{m=-2}^{2} P_{|S=2,m\rangle}- \frac{2}{3}\openone,
\end{equation}
where we have defined the projector $P_{|\psi\rangle}\equiv |\psi\rangle\langle\psi|$ associated with the state $|\psi\rangle$, $|S=2,m\rangle$ denotes the eigenbasis of the joint spin-2 states for neighboring sites $i$ and $i+1$, and $\openone$ is the identity operator in the spin-2 subspace.

For the quartic Hamiltonian, it is seen by straightforward calculation that
\begin{equation}
(H_q)_{i,i+1}=P_{ (|S=2,2\rangle+|S=2,-2\rangle)/\sqrt{2}}\,+P_{|S=2,0\rangle}+\frac{2}{3}\openone.
\end{equation}
For the cubic Hamiltonian, it is seen that
\begin{equation}
(H_c)_{i,i+1}=2\sqrt{3}\Big(P_{|\phi^+\rangle}-P_{|\phi^-\rangle}\Big),
\end{equation}
where 
\begin{eqnarray}
|\phi^\pm\rangle&\equiv & (|S=2,m=2\rangle+|S=2,m=-2\rangle \nonumber \\
&&\pm i \sqrt{2}|S=2,m=0\rangle)/2.
\end{eqnarray}

Since the AKLT state is annihilated by any  spin-2 projectors, it will remain the ground state if the following operator is positive,
\begin{eqnarray}
h(\lambda,\mu)&\equiv& 2 P_{|S=2,m=2\rangle}+2 P_{|S=2,m=-2\rangle}+2 P_{|S=2,m=0\rangle} \nonumber\\
&&+ 2\sqrt{3}\lambda\Big(P_{|\phi^+\rangle}-P_{|\phi^-\rangle}\Big) + \mu \, P_{|S=2,m=0\rangle} \nonumber\\
&& +  \mu\, P_{ (|S=2,m=2\rangle+|S=2,m=-2\rangle)/\sqrt{2}},
\end{eqnarray}
which, in the basis of $|S=2,m=\{\pm 2, 0\}\rangle$ is the following $3\times3$ matrix,
\begin{equation}
h(\lambda,\mu)=\left(\begin{array}{ccc}
2+{\mu}/{2}&{\mu}/{2} & -i\sqrt{6}\lambda\\
{\mu}/{2} & 2 +{\mu}/{2} & -i\sqrt{6}\lambda\\
i\sqrt{6}\lambda & i \sqrt{6}\lambda & 2+\mu \end{array}\right).
\end{equation}
By direct diagonalization, we find that the matrix $h(\lambda,\mu)$ is non-negative when $\mu \pm 2\sqrt{3}\lambda +2 > 0$ which indeed gives the region of the AKLT in Fig.~\ref{fig:Fidelity}. 
\section{Summary}\label{sec:summary}
We have presented a straightforward and general formalism for investigating the structure of a wavefunction as constrained (or protected) by a discrete symmetry  group. The wavefunction is organized into two parts: (1) a CG part, whose form is inferred from the symmetry group and (2) a part whose form is not constrained by the symmetry. From the viewpoint of measurement-based quantum computation, one can then use this formalism to discuss whether the ground state of an SPT phase protected by a given symmetry group allows protected gate operations. This happens when, for example, the MPS matrices $A^i$ decompose into the form $A^i = B_i \otimes \sigma_i$ $i.e.$ the virtual vector space decomposes into junk and protected parts.  Generically speaking, the identity gate is not necessarily protected in an arbitrary SPT phase. With the new formalism, we recovered the results of the $\ztzt$ case previously obtained in Ref.~\cite{bartlett2012} and obtain the MPS forms for several other groups. We show that $A_4$ and $S_4$ groups also allow protected identity gate operation. We also constructed a Hamiltonian with $A_4$ and inversion symmetry and found that in an extended region of a two-parameter space, the ground state is exactly the AKLT state. Using the formulation developed here, further exploration of 1D SPT phases and gate protection can be made with arbitrary finite groups. The MPS forms can also allow the study of the properties of 1D SPT phases which would be of interest to the condensed matter community.

Despite the search we still have not identified a 1D SPT phase that generically supports arbitrary universal single-qubit gates, contrary to what we had hoped for. The only gate that can be naturally protected generically in an entire SPT phase is the identity gate. Although not satisfying in terms of quantum computation, it is useful in terms of transmitting quantum information over long distances. The buffering technique recnetly invented in Ref.~\cite{MillerMiyake14} seems necessary to bring forth universal gates, as demonstrated for the $S_4$ symmetry, and it would be interesting to know whether this can be applied generically to all SPT phases with the identity gate known to be protected. Another open question that naturally arises is whether there exists a 2D SPT phase where all ground states in the phase support protected universal quantum computation.
\label{sec:summary}

\medskip \noindent {\bf Acknowledgement.} The authors would like to thank Naveen Prabhakar and Robert Raussendorf for several helpful discussions, and especially Akimasa Miyake for pointing out a mistake in an earlier draft, where the inversion symmetry was incorrectly imposed. This work was supported by the
National Science Foundation under Grants No. PHY 1314748 and No. PHY
1333903.

\bibliography{SPT_Tensor}

\appendix

\section{Some remarks on projective representations} \label{app:proj}
A projective representation respects group multiplication up to a complex phase i.e.,
\begin{equation}
V(g_1) V(g_2) = \omega(g_1,g_2) V(g_1 g_2),
\end{equation}
Group associativity places constrains the phase $\omega$:
\begin{eqnarray}  \label{eq:cocycle_condition}
V(g_1)(V(g_2)V(g_3)) &=& (V(g_1)V(g_2))V(g_3)   \\
\implies \omega(g_1,g_2 g_3)\omega(g_2,g_3) &=& \omega(g_1,g_2)\omega(g_1 g_2,g_3).
\end{eqnarray}
The possible $\omega$'s fall into different classes and each class has its own set of irreducible representations. These classes are labelled by the elements of the second cohomology group of $G$, $H^2(G,\mathbb{C})$. The identity element of this group labels the familiar set of linear irreducible representations.

We also note that there exists a gauge freedom that preserves the equivalence class of $\omega$: If we re-phase $\omega$ as
\begin{equation}
  \label{eq:cocycle_gauge}
  \omega(g_1,g_2) \mapsto \tilde{\omega}(g_1,g_2) =  \omega(g_1,g_2) \frac{\beta(g_1)\beta(g_2)}{\beta(g_1 g_2)},
\end{equation}
 it still satisfies the condition~(\ref{eq:cocycle_condition}) and $\tilde{\omega} \sim \omega$. This means that while the re-phasing transforms each element of the projective representation as $V(g)\mapsto \tilde{V}(g) = \beta(g) V(g)$, the new $\tilde{V}$ belongs to the same class of projective representations as $V$ $i.e.$ $\tilde{V}(g_1).\tilde{V}(g_2) = \tilde{\omega}(g_1,g_2) \tilde{V}(g_1.g_2)$.

\section{SPT phases with spatial-inversion and time-reversal invariance}\label{app:parity+time}
The action of spatial-inversion or parity, $\hat{P}$ can be effected by a combination of an on-site action by some unitary operator $w$ and a reflection, $\hat{R}$ that exchanges lattice sites $n$ and $-n$. 
\begin{equation}
\hat{P} = w_1 \otimes w_2  \cdots \otimes w_N~\hat{R}.
\end{equation}
Since we cannot talk about inversion in disordered systems, we assume the system also has lattice translation invariance. If parity is a symmetry of a wavefunction $\ket{\psi}$, we have
\begin{equation}\label{eq:wavefuction_parity}
\hat{P} \ket{\psi} = \alpha(P)^N \ket{\psi}.
\end{equation}
The condition Eq.~(\ref{eq:wavefuction_parity}) can also be imposed on the level of the MPS matrices that describe $\ket{\psi}$:
\begin{equation}\label{eq:mps_parity}
w_{i j } (A^{j})^T = \alpha(P) N^{-1} A^i N,
\end{equation}
where, $\alpha(P) = \pm 1$ labels parity even or odd and the action of parity on the virtual space $N$ has the property $N^T = \beta(P) N = \pm N$. $\{\alpha(P), \beta(P)\}$ label the 4 distinct phases protected by parity~\cite{top1d_wennew}.

The anti-unitary action of time-reversal $\hat{T}$ on the other hand is effected by a combination of an on-site unitary action acting on the internal spin degrees of freedom, $v$ and a complex conjugation $\hat{K}$: $\hat{T} = v_1 \otimes v_2  \cdots \otimes v_n~\hat{K}$. If this is a symmetry of a wavefunction $\ket{\psi}$, we have
\begin{equation}\label{eq:wavefuction_time}
\hat{T} \ket{\psi} = \ket{\psi}.
\end{equation}
We no longer need to allow an overall phase $\alpha(T)^N$ because of the anti-unitary nature of $\hat{T}$ that allows it to be absorbed into redefining each basis $\ket{i}\rightarrow \sqrt{\alpha(T)} \ket{i}$.
\begin{proof}
	\begin{eqnarray}
	\hat{T} \ket{\psi} &=& \alpha(T)^N \ket{\psi},\\
	\sqrt{\alpha^*(T)^{N}} \hat{T} \ket{\psi} &=& \sqrt{\alpha(T)^N }\ket{\psi},\\
	\hat{T} \sqrt{\alpha(T)^{N}} \ket{\psi} &=& \sqrt{\alpha(T)^N }\ket{\psi},\\
	\hat{T} \ket{\psi'} &=& \ket{\psi'}.
	\end{eqnarray}
\end{proof}

 The condition of Eq.~(\ref{eq:wavefuction_parity}) can also be imposed on the level of the MPS matrices that describes $\ket{\psi}$
\begin{equation}\label{eq:mps_timereversal}
v_{i j } (A^{j})^* =  M^{-1} A^i M,
\end{equation}
where, $M M^* = \beta(T) \mathbb{1} = \pm \mathbb{1}$ and $\beta(T)$ labels the two distinct phases of time-reversal invariant Hamiltonians. 

Finally if we consider systems invariant under both parity and time-reversal, there are 8 distinct phases labelled by $\{ \alpha(P), \beta(P), \beta(T) \}$ as defined before. However, since the action of parity and time-reversal should commute, this imposes constraints on the matrices $M$ and $N$ as
\begin{equation}\label{eq:parity+timereversal}
MN^\dagger M N^\dagger \propto \ e^{i \theta} \mathbb{1}
\end{equation}
We direct the reader to Ref.~\cite{top1d_wennew,top1d_wennew} for details on several results used in this section. 

\section{SPT Phases with combination of on-site symmetry, spatial-inversion and time-reversal invariance}\label{app:parity+time+onsite}
We now look at ground states of SPT phases of gapped Hamiltonians with on-site symmetry combined with parity, time-reversal invariance or both. We find that the `B' matrices of decomposition of Sec.~\ref{sec:on-site} have further constraints in the way described in Sec.~\ref{app:parity+time}.

\subsubsection{On-site symmetry + parity}
Let us consider SPT phases protected by an on-site symmetry $G$ under a representation $u(g)$ combined with parity. If the actions of the two symmetry transformations commute on the physical level,
\begin{equation}
\hat{U} \hat{P} \ket{\psi} = \hat{P} \hat{U} \ket{\psi},
\end{equation}
this imposes constraints on the matrix $N$ defined in Sec.~\ref{app:parity+time} as~\cite{top1d_wennew}.
\begin{equation}\label{eq:on-site+parity_MPS}
N^{-1} V(g) N = \gamma(g) V^*(g).
\end{equation}
Where, $\gamma(g)$ is a one-dimensional irrep of $G$ that arises from the commutation of on-site and parity transformations~\cite{top1d_wennew} and $V(g)$ is the reduced (block-diagonal) representation of $G$ acting on the virtual space as discussed in Sec~\ref{sec:on-site} that contains all the irreps, $V_1 \cdots V_r$ of a certain projective class $\omega$.
\begin{equation} \label{eq:V(g) form}
V(g) = \begin{pmatrix}
\mathbb{1}_1 \otimes V_1(g) & 0 & \dots & 0\\
0 &\mathbb{1}_2 \otimes V_2(g)   & \dots & \vdots\\
\vdots &  & \ddots & \vdots\\
0 & \dots & 0 & \mathbb{1}_r \otimes V_r(g) 
\end{pmatrix}.
\end{equation}
$\mathbb{1}_i$ is the trivial action on the degeneracy space of `B' matrices as defined earlier. Different phases of matter are now labeled by $\{\omega,\chi(g),\alpha(P),\beta(P),\gamma(g) \}$~\cite{top1d_wennew}:
\begin{enumerate}
	\item The different projective classes $\omega \in H^2(G,\mathbb{C})$ which satisfy $\omega^2 = 1$. 
	\item The different one-dimensional irreps $\chi$ of $G$ since the system is translationally invariant.
	\item $\alpha(P)$, the parity of the spin chain
	\item $\beta(P)$ which denotes whether the virtual space parity representation is symmetric or anti-symmetric. 
	\item $\gamma(g) \in \mathcal{G}/\mathcal{G}_2$ where $\mathcal{G}$ labels the group of 1D irreps of $G$ and $\mathcal{G}_2$ labels the group of the square of 1D irreps of $G$. This arises due to the commutation of parity and on-site symmetry transformations in the virtual space. 
\end{enumerate}
Given a set of labels, $\{\omega,\chi(g),\alpha(P),\beta(P),\gamma(g) \}$, we constrain the MPS ground-state wavefunction We observe that the right hand side of Eq.~(\ref{eq:on-site+parity_MPS}) can be written as
\begin{equation}
\gamma(g) V^*(g) = L_\gamma V(g) L_\gamma^{-1},
\end{equation}
where, $L_\gamma$ involves permutation of irrep blocks and possibly a change of basis on the irreps of $V(g)$ and can be obtained by considering the effect of re-phasing each of the complex conjugated irrep blocks $V^*_\alpha(g)$ with $\gamma(g)$. 
\begin{proof}
	To see this, we first note that when $\omega^2 =1$ $i.e.$ $\omega = \omega^*$, $V^*_\alpha(g)$ is a representation that belongs to the same class of projective irreps $\omega$ as $V_\alpha(g)$ as seen by complex conjugating Eq.~(\ref{eq:projective_defn}).  $\gamma(g) V^*_\alpha(g)$ also belongs to the same class because $\gamma(g)$ belongs to the class labelled by the trivial element $e \in \hgc$ and hence $\gamma(g) V^*_\alpha(g)$ belongs to the class $e* \omega = \omega$. To show that $\gamma(g) V^*_\alpha(g)$ is also an irrep, we start with the characters $\chi_\alpha$ of the irrep $V_\alpha$ which satisfy the irrep condition of the group of order $|G|$~\cite{ramond_group}
	\begin{equation}
	\frac{1}{|G|} \sum_{g \in G} \chi(g) \chi^*(g) = 1
	\end{equation}
	The characters of $\gamma(g) V^*_\alpha(g)$, $\bar{\chi}_\alpha = \gamma \chi^*_\alpha$ can also easily be shown to satisfy the same condition
	\begin{multline}
	\frac{1}{|G|} \sum_{g \in G} \bar{\chi}(g) \bar{\chi}^*(g) = \\
	\frac{1}{|G|} \sum_{g \in G} \gamma(g) \chi^*(g) \gamma^*(g) \chi(g)= 1
	\end{multline}
	Thus $\gamma(g) V^*_\alpha(g) \sim V_{p(\alpha )}(g)$ is some other irrep in the class $\omega \in \hgc$. We can check that $V_{p(\alpha )}$ again form the complete set of irreps as we run over $\alpha$. This means that the reduced representation $\gamma(g) V^*(g)$ can be obtained from Eq.~(\ref{eq:V(g) form}) by permuting the irrep blocks and with a change of basis and can be done using a matrix $L_\gamma$.
	\begin{equation}
	\gamma(g) V^*(g) = L_\gamma V(g) L_\gamma^{-1}
	\end{equation}
\end{proof}

Using this, Eq.~(\ref{eq:on-site+parity_MPS}) can be rewritten as
\begin{equation}\label{eq:on-site+parity_schurform}
(N L_\gamma)^{-1} V(g) (N L_\gamma) = V(g).
\end{equation}
Eq.~(\ref{eq:on-site+parity_schurform}) imposes constraints on the matrix $N L_\gamma$ block-wise using Schur's lemma for each irrep block of $V(g)$, 
\begin{equation}
N L_\gamma = \begin{pmatrix}\label{eq:N_blockcondition}
N_1 \otimes \mathbb{1'}_1  & 0 & \dots & 0\\
0 & N_2 \otimes \mathbb{1'}_2  & \dots & \vdots\\
\vdots &  & \ddots & \vdots\\
0 & \dots & 0 & N_r \otimes \mathbb{1'}_r 
\end{pmatrix},
\end{equation}
where $\mathbb{1'}_\alpha$ is the identity matrix in the irrep $V_\alpha$. Moving $L_\gamma$ to the other side of the equation gives the form of $N$. This form can be used in the condition Eq.~(\ref{eq:mps_parity}) which effectively results in conditions of the `B' matrices of $A^i$ of Eq.~(\ref{eq:mps_onsite+trans}) determined from labels $\{ \omega, \chi \}$. So far, we have used the labels $\{ \omega, \chi(g), \gamma(g) \}$ to constrain the MPS matrices. The labels $\alpha(P)$ and $\beta(P)$ determine the form of the blocks $N_\alpha$ and are imposed on the `B' matrices when we use Eq.~(\ref{eq:mps_parity}) and the results of Sec.~\ref{app:parity+time}.

\subsubsection{On-site symmetry + time reversal}
We can repeat the same exercise for time-reversal invariance combined with on-site symmetry $G$. If the actions of the two symmetry transformations commute
\begin{equation}
\hat{U} \hat{T} \ket{\psi} = \hat{T} \hat{U} \ket{\psi},
\end{equation}
We find that the condition on the matrix $M$ that results is identical to Eq.~(\ref{eq:on-site+parity_MPS})~\cite{top1d_wennew}.
\begin{equation}\label{eq:on-site+timereversal_MPS}
M^{-1} V(g) M = \gamma'(g) V^*(g)
\end{equation}
With additional translation invariance, different SPT phases are labelled by $\{\omega,\chi(g),\beta(T),\gamma'(g) \}$~\cite{top1d_wennew} $i.e.$
\begin{enumerate}
	\item The different projective classes $\omega \in H^2(G,\mathbb{C})$ which satisfy $\omega^2 = 1$. 
	\item The different one-dimensional irreps $\chi$ of $G$ which satisfy $\chi^2=1$ if the system is translationally invariant. If not, different $\chi$ all label the same phase.
	\item $\beta(T)$ defined by $M M^* = \beta(T) \mathbb{1}$
	\item $\gamma'(g) \in \mathcal{G}/\mathcal{G}_2$ where $\mathcal{G}$ labels the group of 1D irreps of $G$ and $\mathcal{G}_2$ labels the group of the square of 1D irreps of $G$. This arises due to the commutation of time-reversal and on-site symmetry transformations in the virtual space. 
\end{enumerate}
In the same way as for parity, we can find $L_{\gamma'}$ and the condition on $M$
\begin{equation}\label{eq:M_blockcondition}
M L_{\gamma'} = \begin{pmatrix}
M_1 \otimes \mathbb{1'}_1  & 0 & \dots & 0\\
0 & M_2 \otimes \mathbb{1'}_2 & \dots & \vdots\\
\vdots &  & \ddots & \vdots\\
0 & \dots & 0 & M_r \otimes \mathbb{1'}_r 
\end{pmatrix}
\end{equation}
Moving $L_{\gamma'}$ to the right hand side, we get the form of $M$ and can use this in Eq.~(\ref{eq:mps_timereversal}) to constrain the `B' matrices of $A^i$ in Eq.~(\ref{eq:mps_onsite+trans}) employing labels $\{ \omega, \chi(g), \gamma'(g) \}$ thus far. The label $\beta(T)$ determines the form of the blocks $M_\alpha$ and is imposed on the `B' matrices when we use Eq.~(\ref{eq:mps_timereversal}) and the results of Sec.~\ref{app:parity+time}.

\subsubsection{On-site symmetry + parity + time reversal}
Finally, we consider the combined action of on-site symmetry, spatial-inversion and time-reversal invariance. The distinct SPT phases are labelled by $\{\omega, \chi(g), \alpha(P), \beta(P),  \beta(T),\gamma(g), \gamma'(g) \}$ where all labels are defined as before with additional conditions $\omega^2 = 1$ and $\chi^2 = 1$~\cite{top1d_wennew}. To write down the MPS form for the ground state of a phase labelled by these labels, we repeat the same procedure as we did before and obtain the forms of $L_\gamma$ and $L_{\gamma'}$. Using this, we constrain the block form of $M$, $N$ using Eqs.~(\ref{eq:N_blockcondition},\ref{eq:M_blockcondition}). The blocks of $M$ and $N$ encode the information about $\{\alpha(P),\beta(P),\beta(T)\}$ and are used to constrain the `B' matrix blocks of $A^i$ in Eq.~(\ref{eq:mps_onsite+trans}) using Eqs.~(\ref{eq:mps_parity},\ref{eq:mps_timereversal}).

We summarize this section with steps used to constrain ground states of SPT phases of Hamiltonians invariant under combinations of on-site symmetry with parity and/or time reversal:
\begin{enumerate}
	\item The different SPT phases are labelled by a subset of the following labels  $\{\omega, \chi(g), \alpha(P), \beta(P),  \beta(T),\gamma(g), \gamma'(g) \}$ with $\omega^2 = 1$ and $\chi^2 = 1$.
	\item Impose the labels from on-site symmetry $i.e.$ $\{\omega, \chi(g)\}$ using the steps of Sec~\ref{sec:on-site+trans}).
	\item Impose the label $\gamma$ ($\gamma'$) from parity (time-reversal) symmetry by constructing $L_\gamma$ ($L_{\gamma'}$) and thus constraining the matrices $N$ ($M$) to a block form using Eqs.~(\ref{eq:N_blockcondition},\ref{eq:M_blockcondition})). 
	\item Impose labels $\{\alpha(P),\beta(P),\beta(T)\}$ by restricting the form of the blocks of $N$, $M$ appropriately and then using Eqs.~(\ref{eq:mps_parity},\ref{eq:mps_timereversal}).
\end{enumerate} 
We remark that while we can use $L_\gamma, L_{\gamma'}$ to determine the block form of $M$ and $N$, constraining the individual blocks themselves is not straightforward and we do not investigate a way to do it in this paper.

\section{Obtaining the Clebsch-Gordan coefficients}
\label{app:sakata}
We now review a method to obtain the CG matrices corresponding to finite group irrep decompositions of a certain kind. We follow the technique developed in Ref.~\cite{sakata}. 
Essentially what is needed are the two theorems presented below.
\begin{theorem}\label{th:eqvt}
Consider a finite group $G$ and a certain irrep $D(r)$, $r \in G$. If $D'(r)$ is an equivalent irrep $i.e.$ $D'(g) = U D(g) U^{\dagger}$ then $\sum_{r \in G} D'(r) A D^{\dagger}(r) = \lambda U  $ where $A$ is an arbitrary matrix which is of the same size as $D$ and $\lambda$ is a constant which is a function of the elements of $A$
\end{theorem}

To prove Theorem 1, we need the following two lemmas.
  \begin{lemma}
   M = $ \sum_{r \in G} D(r) B \dagr{D(r)} \propto \mathbb{1} $ where $B$ is an arbitrary matrix of the same size as $D$.
  \end{lemma}
   \begin{proof}
     \begin{multline}
       D(g) M = D(g) \sum_{r \in G} D(r) B \dagr{D(r)} \\ =  \sum_{r \in G} D(g)D(r) B \dagr{D(r)}  
       =\sum_{r \in G} D(g r) B \dagr{D(r)}\\ = \sum_{g r \in G} D(gr) B \dagr{D(gr)} D(g) = M D(g),
     \end{multline}
     \begin{equation}
       \implies [M,D(g)] = 0~ \forall g \in G.
     \end{equation}
     From Schur's second lemma, we get $M \propto \mathbb{1}$ 
   \end{proof}

\begin{lemma}
  If $D^\alpha (g)$ and $D^\beta (g)$ are two inequivalent irreps, $M' = \sum_{r \in G} D^\alpha (r) B \dagr{D^\beta (r)} = 0$
\end{lemma}
\begin{proof}
  Using the same arguments as before, we get $D^\alpha (g) M' = M' D^\beta (g)$. From Schur's first lemma we get $M' = 0$
\end{proof}

To prove theorem \ref{th:eqvt}, let us start with
   \begin{equation}
 \sum_{r \in G} D(r) B \dagr{D(r)} =  \lambda \mathbb{1}.      
   \end{equation}
Then take $B = \dagr{U} A $, we have
\begin{multline}
   \sum_{r \in G} D(r) \dagr{U} A \dagr{D(r)} =  \lambda \mathbb{1} \\ \implies  \sum_{r \in G}U D(r) \dagr{U} A \dagr{D(r)}  =  \lambda U \\ \implies \sum_{r \in G} D'(r) A \dagr{D(r)} =  \lambda U.
\end{multline}

\begin{theorem} \label{th:ineqvt}
  Let $D^\alpha(g)$ and $D^\beta(g)$ be two irreps of $G$. Let $D'(g) = D^\alpha(g) \otimes D^\beta(g)$ be the direct product representation of irreps whose CG decomposition is multiplicity free $i.e.$ $\alpha \otimes \beta=\oplus_\gamma n_{\alpha \beta}^\gamma \gamma$ has all $n_{\alpha \beta}^\gamma \le 1$. Let $D(g)$ be the completely reduced representation which is block diagonal containing all irreps in the decomposition of $\alpha \otimes \beta$ labelled  $\gamma = 1 \ldots m$.
  \begin{equation}
D(g)=
    \begin{pmatrix}
      D_{1}(g) & 0 & \dots & 0\\
      0 & D_{2}(g) & \dots & \vdots\\
      \vdots &  & \ddots & \vdots\\
      0 & \dots & 0 & D_{m}(g) 
    \end{pmatrix}.
  \end{equation}
If $U$ consists of the CG matrices such that $D'(r) = U D(r) \dagr{U}$, organized according to the irrep sizes,
\begin{equation}
U=
  \begin{pmatrix}
    U_{11} & U_{1 2} & \dots & U_{1 m} \\
    U_{21} & U_{2 2} & \dots & U_{2 m} \\
    \vdots & \ddots & & \vdots\\
    \vdots & & \ddots & \vdots\\
    U_{m1} & U_{m 2} & \dots & U_{m m} \\
  \end{pmatrix},
\end{equation}
then
\begin{equation}
  \sum_{r \in G} D'(r) A \dagr{D(r)} =
  \begin{pmatrix}
    \lambda_1 U_{11} & \lambda_2 U_{1 2} & \dots & \lambda_m U_{1 m} \\
    \lambda_1 U_{21} & \lambda_2 U_{2 2} & \dots & \lambda_m U_{2 m} \\
    \vdots & \ddots & & \vdots\\
    \vdots & & \ddots & \vdots\\
    \lambda_1 U_{m1} & \lambda_2 U_{m 2} & \dots & \lambda_m U_{m m}     
  \end{pmatrix}.
\end{equation}
\end{theorem}
We need the following Lemma to prove Theorem 2.
  \begin{lemma}
    \begin{equation}
    \sum_{r \in G} D(r) B \dagr{D(r)}=
    \begin{pmatrix}
      \lambda_1 \mathbb{1}_{1} & 0 & \dots & 0\\
      0 &  \lambda_2 \mathbb{1}_{2} & \dots &0\\
      \vdots & \ddots & & \vdots\\
      \vdots & & \ddots  & \vdots\\\
      0 & \dots & &\lambda_m \mathbb{1}_{m }
    \end{pmatrix}.
    \end{equation}
  \end{lemma}
  \begin{proof}
    \begin{multline}
      \sum_{r \in G} D(r) B \dagr{D(r)}= \\ \sum_{r } 
     \begin{pmatrix}
         D_{1}(r) B_{11} \dagr{D_{1}(r)} &  \dots &  D_{1}(r) B_{1m} \dagr{D_{m}(r)}\\
        \vdots & \ddots &  \vdots\\
         D_{m}(r) B_{m1} \dagr{D_{1}(r)}  & \dots &  D_{m}(r) B_{mm} \dagr{D_{m}(r)}
      \end{pmatrix}.
    \end{multline}
Using the results of the last two Lemmas, we get 
    \begin{equation}
    \sum_{r \in G} D(r) B \dagr{D(r)}=
    \begin{pmatrix}
      \lambda_1 \mathbb{1}_{1} & 0 & \dots & 0\\
      0 &  \lambda_2 \mathbb{1}_{2} & \dots &0\\
      \vdots & \ddots & & \vdots\\
      \vdots & & \ddots  & \vdots\\\
      0 & \dots & &\lambda_m \mathbb{1}_{m}
    \end{pmatrix}.
    \end{equation}
  \end{proof}
To prove Theorem 2, we once again take $B = \dagr{U} A$,  and thus
\begin{multline}
\sum_{r \in G} U D(r) B \dagr{D(r)}=  \sum_{r \in G} D'(r) A \dagr{D(r)} =\\
  \begin{pmatrix}
    \lambda_1 U_{11} & \lambda_2 U_{1 2} & \dots & \lambda_m U_{1 m} \\
    \lambda_1 U_{21} & \lambda_2 U_{2 2} & \dots & \lambda_m U_{2 m} \\
    \vdots & \ddots & & \vdots\\
    \vdots & & \ddots & \vdots\\
    \lambda_1 U_{m1} & \lambda_2 U_{m 2} & \dots & \lambda_m U_{m m}     
  \end{pmatrix}.
\end{multline}
Thus, normalizing $\sum_{r \in G} D'(r) A \dagr{D(r)}$ appropriately gives us all the required CG matrices up to multiplication by a complex number. This ambiguity gets absorbed into the `B' matrices when we use the CG coefficients to write down MPS matrices.

We note that for the groups $\ztzt$, $D_4$ and $A_4$, when we take a direct product of the irreps of the physical spin with any projective irrep, we get a multiplicity-free CG decomposition for which we can use the technique mentioned above to obtain CG coefficients. However, for the case of $S_4$, the irrep of the physical spin $3_{(1)}$ has the following decomposition when we take the direct product with the projective irrep $\tilde{4}:$ $3_{(1)}\otimes\tilde{4} = \tilde{2}_{(0)} \oplus\ \tilde{2}_{(1)} \oplus \tilde{4} \oplus \tilde{4}$. Clearly, $\tilde{4}$ has multiplicity 2 in the decomposition. In this case, if we apply the procedure above nonetheless, we get the following:
\begin{multline}
 \sum_{r \in G} D'(r) A \dagr{D(r)}  = \\ \begin{pmatrix}
\begin{pmatrix}
 ~\\
 \lambda_1 C_{3_{(1)} \tilde{4}}^{\tilde{2}_{(0)}}\\
 ~
 \end{pmatrix} &
 \begin{pmatrix}
  ~\\
  \lambda_2 C_{3_{(1)} \tilde{4}}^{\tilde{2}_{(1)}}\\
  ~
  \end{pmatrix}   &
  \begin{pmatrix}
   \lambda_3 C_{3_{(1)} \tilde{4}}^{\tilde{4};1}\\
    +\\
  \lambda_4 C_{3_{(1)} \tilde{4}}^{\tilde{4};2}
  \end{pmatrix} &
  \begin{pmatrix}
  \mu_3 C_{3_{(1)} \tilde{4}}^{\tilde{4};1}\\
  +\\
  \mu_4 C_{3_{(1)} \tilde{4}}^{\tilde{4};2}
  \end{pmatrix}  
 	\end{pmatrix}
\end{multline}
Where $D'(g) = D_{3_{(1)}} \otimes D_{\tilde{4}}$, $D(g) = D_{\tilde{2}_{(0)}} \oplus D_{\tilde{2}_{(1)}} \oplus D_{\tilde{4}} \oplus D_{\tilde{4}}$ and the $C_{3_{(1)} \tilde{4}}^{\tilde{2}_{(1)}}$ etc represent blocks of CG coefficients with the $m$ labels suppressed. 

We can see that $C_{3_{(1)} \tilde{4}}^{\tilde{4};1}$ and $C_{3_{(1)} \tilde{4}}^{\tilde{4};2}$ cannot in principle be separated which is why the method fails for decompositions with irrep multiplicities. However, in our case, it so happens that because of a convenient block structure we can separate the matrices by hand and obtain all CG coefficients.

\end{document}